\documentclass[aps,preprint,nofootinbib,superscriptaddress]{revtex4-1}

\usepackage{amsmath,amssymb,amsfonts}
\usepackage{amsthm}
\usepackage[utf8]{inputenc}
\usepackage{graphicx}
\usepackage{url}

\usepackage[bookmarks,colorlinks,breaklinks]{hyperref}
\hypersetup{linkcolor=blue,citecolor=blue,filecolor=dullmagenta,urlcolor=blue}

\usepackage{bm,bbm}

\newcommand{\nwse}[3]{\ensuremath{#1^{#2}_{\phantom{#2} #3}}}
\newcommand{\swne}[3]{\ensuremath{#1_{#2}^{\phantom{#2} #3}}}
\newcommand{\Tr}{\ensuremath{\text{Tr}}}

\newtheorem{theorem}{Theorem}[section]

\newcommand{\ucsc}{Departamento de Matemática y Física Aplicadas, Universidad Católica de la Santísima Concepción, Alonso de Ribera 2850, 4090541 Concepción, Chile}
\newcommand{\udec}{Departamento de F\'{\i}sica, Universidad de Concepci\'{o}n, Casilla 160-C, Concepci\'{o}n, Chile}
\newcommand{\torino}{Dipartimento di Fisica, Politecnico di Torino, C.so Duca degli Abruzzi, 24, I-10129 Torino, Italy}
\newcommand{\unap}{Departamento de F\'{\i}sica y Matem\'{a}ticas, Universidad Arturo Prat, Casilla 121, Iquique, Chile}

\begin{document}

\title{A generalized action for $\left( 2 + 1 \right)$-dimensional Chern--Simons gravity}

\author{Jos\'{e} D\'{\i}az}
\email{jose.diaz.polanco@unap.cl}
\affiliation{\unap}

\author{Octavio Fierro}
\email{ofierro@udec.cl}
\affiliation{\udec}

\author{Fernando Izaurieta}
\email{fizaurie@ucsc.cl}
\affiliation{\ucsc}

\author{Nelson Merino}
\email{nemerino@udec.cl}
\affiliation{\udec}
\affiliation{\torino}

\author{Eduardo Rodr\'{\i}guez}
\email{edurodriguez@ucsc.cl}
\affiliation{\ucsc}

\author{Patricio Salgado}
\email{pasalgad@udec.cl}
\affiliation{\udec}

\author{Omar Valdivia}
\email{ovaldivi@udec.cl}
\affiliation{\unap}
\affiliation{\udec}

\date{June 7, 2012}

\begin{abstract}
We show that the so-called semi-simple extended Poincar\'{e} (SSEP) algebra in $D$
dimensions can be obtained from the anti-de~Sitter algebra
$\mathfrak{so} \left( D-1,2 \right)$
by means of the $S$-expansion procedure with an appropriate semigroup $S$.
A general prescription is given for computing Casimir operators for $S$-expanded algebras,
and the method is exemplified for the SSEP algebra.
The $S$-expansion method also allows us to extract the
corresponding invariant tensor for the SSEP algebra,
which is a key ingredient in the construction of a generalized
action for Chern--Simons gravity in $2+1$ dimensions.
\end{abstract}

\maketitle

\section{Introduction}

In Refs.~\cite{Sor04,Dup05,Sor06,Sor10}, the Poincar\'{e} algebra of rotations $\bm{J}_{ab}$ and
translations $\bm{P}_{a}$ in $D$-dimensional spacetime has been extended by the
inclusion of the second-rank tensor generator $\bm{Z}_{ab}$ in the following way: 
\begin{align}
\left[ \bm{J}_{ab}, \bm{J}_{cd} \right] & = \eta_{ad} \bm{J}_{bc} + \eta_{bc} \bm{J}_{ad} -
\eta_{ac} \bm{J}_{bd} - \eta_{bd} \bm{J}_{ac},  \label{eext1} \\
\left[ \bm{J}_{ab}, \bm{P}_{c} \right] & = \eta_{bc} \bm{P}_{a} - \eta_{ac} \bm{P}_{b}, \\
\left[ \bm{P}_{a}, \bm{P}_{b} \right] & = c \bm{Z}_{ab}, \\
\left[ \bm{J}_{ab}, \bm{Z}_{cd} \right] & = \eta_{ad} \bm{Z}_{bc} + \eta_{bc} \bm{Z}_{ad} -
\eta_{ac} \bm{Z}_{bd} - \eta_{bd} \bm{Z}_{ac}, \\
\left[ \bm{Z}_{ab}, \bm{P}_{c} \right] & = \frac{4a^{2}}{c} \left( \eta_{bc} \bm{P}_{a} -
\eta_{ac} \bm{P}_{b} \right), \\
\left[ \bm{Z}_{ab}, \bm{Z}_{cd} \right] & = \frac{4a^{2}}{c} \left[ \eta_{ad} \bm{Z}_{bc} +
\eta_{bc} \bm{Z}_{ad} - \eta_{ac} \bm{Z}_{bd} - \eta_{bd} \bm{Z}_{ac} \right],  \label{2.1}
\end{align}
where $a$ and $c$ are constants.
It is remarkable that the Lie algebra~(\ref{eext1})--(\ref{2.1}) is semi-simple,
in contrast to the Poincar\'{e} and extended Poincar\'{e} algebras [cf.~eqs.~(1.1) and~(1.2) of Ref.~\cite{Sor06}].
Note that, in the $a \rightarrow 0$ limit, the algebra~(\ref{eext1})--(\ref{2.1})
reduces to the algebra in eq.~(1.2) of Ref.~\cite{Sor06}.
The \emph{semi-simple} extended Poincar\'{e} (SSEP) algebra~(\ref{eext1})--(\ref{2.1}) can be
rewritten in the form 
\begin{align}
\left[ \bm{N}_{ab}, \bm{N}_{cd} \right] & = \eta_{ad} \bm{N}_{bc} + \eta_{bc} \bm{N}_{ad} - \eta
_{ac}\bm{N}_{bd} - \eta_{bd} \bm{N}_{ac},  \label{2.7} \\
\left[ \bm{L}_{AB}, \bm{L}_{CD} \right] & = \eta_{AD} \bm{N}_{BC} + \eta_{BC} \bm{N}_{AD} -
\eta_{AC} \bm{N}_{BD} - \eta_{BD} \bm{N}_{AC},  \label{2.8} \\
\left[ \bm{N}_{ab}, \bm{L}_{CD} \right] & = 0,  \label{2.9}
\end{align}
where the metric tensor $\eta_{AB}$ is given by 
\begin{equation}
\eta_{AB} = \left[ 
\begin{array}{cc}
\eta_{ab} & 0 \\ 
0 & -1%
\end{array}
\right]  \label{2.10}
\end{equation}
and the $\bm{N}_{ab}$ generators read 
\begin{equation}
\bm{N}_{ab} = \bm{J}_{ab} - \frac{c}{4a^{2}} \bm{Z}_{ab}.  \label{2.11}
\end{equation}
The $\bm{N}_{ab}$ generators form a basis for the Lorentz algebra $\mathfrak{so} \left( D-1,1 \right)$.
The $\bm{L}_{AB}$ generators, on the other hand, are given by 
\begin{equation}
\bm{L}_{AB} = \left[ 
\begin{array}{cc}
\bm{L}_{ab} & \bm{L}_{a,D} \\ 
\bm{L}_{D,a} & \bm{L}_{D,D}%
\end{array}
\right] = \left[ 
\begin{array}{cc}
\frac{c}{4a^{2}} \bm{Z}_{ab} & \frac{1}{2a} \bm{P}_{a} \\ 
-\frac{1}{2a} \bm{P}_{a} & 0%
\end{array}
\right]  \label{2.12}
\end{equation}
and form a basis for the anti-de~Sitter (AdS) $\mathfrak{so} \left( D-1,2 \right)$ algebra.
The SSEP algebra (\ref{2.7})--(\ref{2.9}) is thus seen to be the direct sum
$\mathfrak{so} \left( D-1,1 \right) \oplus \mathfrak{so} \left( D-1,2 \right)$
of the $D$-dimensional Lorentz algebra and the $D$-dimensional AdS algebra.

Using (\ref{2.11}) and (\ref{2.12}) in (\ref{2.7})--(\ref{2.9}) we find that
the SSEP algebra~(\ref{eext1})--(\ref{2.1}) can be rewritten as
\begin{align}
\left[ \bm{N}_{ab},\bm{N}_{cd}\right] & =\eta _{ad}\bm{N}_{bc}+\eta _{bc}\bm{N}_{ad}-\eta
_{ac}\bm{N}_{bd}-\eta _{bd}\bm{N}_{ac},  \label{2.12'} \\
\left[ \bm{L}_{ab},\bm{L}_{cd}\right] & =\eta _{ad}\bm{L}_{bc}+\eta _{bc}\bm{L}_{ad}-\eta
_{ac}\bm{L}_{bd}-\eta _{bd}\bm{L}_{ac}, \\
\left[ \bm{L}_{ab},\bm{L}_{c,D}\right] & =\eta _{bc}\bm{L}_{a,D}-\eta _{ac}\bm{L}_{b,D}, \\
\left[ \bm{L}_{a,D},\bm{L}_{c,D}\right] & =\bm{L}_{ac}, \\
\left[ \bm{N}_{ab},\bm{L}_{cd}\right] & =0, \\
\left[ \bm{N}_{ab},\bm{L}_{c,D}\right] & =0.  \label{2.13}
\end{align}

It is the purpose of this Letter to show that the SSEP algebra
$\mathfrak{so} \left( D-1,1 \right) \oplus \mathfrak{so} \left( D-1,2 \right)$
can be obtained from the AdS algebra $\mathfrak{so} \left( D-1,2 \right)$ via the $S$-expansion procedure with an
appropriate semigroup $S$ \cite{Iza06b,Iza09a}.
The $S$-expansion method also allows
us to compute an invariant tensor for the SSEP
algebra, which is a key ingredient in the construction of the more general
action for Chern--Simons (CS) gravity in $2+1$ dimensions.

The article is organized as follows.
In section~\ref{sec:sexpp} we briefly review the main aspects of the $S$-expansion procedure.
In section~\ref{sec:sexpAdS} we cast the SSEP algebra as an $S$-expansion
of the $D$-dimensional AdS algebra
$\mathfrak{so} \left( D-1,2 \right)$ through an appropriate semigroup $S$.
Section~\ref{sec:Casgral} is devoted to a systematic exposition of the construction of Casimir operators for $S$-expanded Lie algebras.
This general procedure is applied in section~\ref{sec:Caspart} to the case of the SSEP algebra.
In section~\ref{sec:CS} we compute an invariant tensor for the SSEP algebra and put it to use by
constructing the more general action for CS gravity in $2+1$ dimensions.
A brief comment in section~\ref{sec:final} concludes the paper.

\section{The $S$-Expansion Procedure}
\label{sec:sexpp}

In this section we briefly review the general abelian semigroup expansion procedure ($S$-expansion for short).
We refer the interested reader to Ref.~\citep{Iza06b} for further details.

Consider a Lie algebra $\mathfrak{g}$ and a finite abelian
semigroup $S = \left\{ \lambda_{\alpha} \right\}$. According to Theorem~3.1
from Ref.~\cite{Iza06b}, the direct product $S \times \mathfrak{g}$ is also a Lie
algebra. Interestingly, there are cases when it is possible to
systematically extract subalgebras from $S \times \mathfrak{g}$. Start by
decomposing $\mathfrak{g}$ in a direct sum of subspaces, as in $\mathfrak{g}
= \bigoplus_{p\in I}V_{p}$, where $I$ is a set of indices. The internal
structure of $\mathfrak{g}$ can be codified through the mapping%
\footnote{Here $2^{I}$ stands for the set of all subsets of $I$.}
$i : I \times I \rightarrow 2^{I}$ according to
$\left[ V_{p}, V_{q} \right] \subset \bigoplus_{r \in i \left( p,q \right)} V_{r}$.
When the semigroup $S$
can be decomposed in subsets $S_{p}$, $S = \bigcup_{p \in I} S_{p}$, such
that they satisfy the ``resonant condition'' $S_{p} \cdot S_{q}
\subset \bigcap_{r \in i \left( p,q \right)} S_{r}$,%
\footnote{Here $S_{p} \cdot S_{q}$ denotes the set of all the products of all elements from $S_{p}$ with all elements from $S_{q}$.}
then we have that
$\mathfrak{G}_{\text{R}}=\bigoplus_{p\in I}S_{p}\times V_{p}$ is a
``resonant subalgebra'' of $S \times \mathfrak{g}$ (see Theorem~4.2 from
Ref.~\cite{Iza06b}).

An even smaller algebra can be obtained when there is a zero element in the
semigroup, i.e., an element $0_{S}\in S$ such that, for all $\lambda_{\alpha
}\in S$, $0_{S}\lambda_{\alpha}=0_{S}$. When this is the case, the whole
$0_{S}\times\mathfrak{g}$ sector can be removed from the resonant subalgebra
by imposing $0_{S}\times\mathfrak{g}=0$. The remaining piece, to which we
refer to as ``$0_{S}$-reduced algebra,'' continues to be a Lie algebra
(see $0_{S}$-reduction and Theorem~6.1 from Ref.~\cite{Iza06b}).

\section{$S$-Expansion of the anti-de~Sitter Algebra}
\label{sec:sexpAdS}

In this section we sketch the steps to be undertaken in order to obtain the SSEP algebra,
$\mathfrak{so} \left( D-1,1 \right) \oplus \mathfrak{so} \left( D-1,2 \right)$,
as an $S$-expansion of the AdS algebra,
$\mathfrak{so} \left( D-1,2 \right)$.

The first step consists of splitting the AdS algebra in subspaces, i.e.,
$\mathfrak{so}\left( D-1,2\right) =V_{0}\oplus V_{1}$,
where $V_{0}$ corresponds to the Lorentz subalgebra
$\mathfrak{so}\left(D-1,1\right)$, which is generated by $\bar{\bm{J}}_{ab}$,
and $V_{1}$ corresponds to the AdS ``boosts,'' generated by $\bar{\bm{P}}_{a}$.
The generators $\bar{\bm{J}}_{ab}$, $\bar{\bm{P}}_{a}$ satisfy the following commutation relations:
\begin{align}
\left[ \bar{\bm{P}}_{a},\bar{\bm{P}}_{b}\right] &=\bar{\bm{J}}_{ab} \\
\left[ \bar{\bm{J}}_{ab},\bar{\bm{P}}_{c}\right] &=\eta _{cb}\bar{\bm{P}}_{a}-\eta _{ca}%
\bar{\bm{P}}_{b} \\
\left[ \bar{\bm{J}}_{ab},\bar{\bm{J}}_{cd}\right] &=\eta _{ad}\bar{\bm{J}}_{bc}+\eta _{bc}%
\bar{\bm{J}}_{ad}-\eta _{ac}\bar{\bm{J}}_{bd}-\eta _{bd}\bar{\bm{J}}_{ac}.
\end{align}
The subspace structure can be written as 
\begin{align}
\left[ V_{0},V_{0}\right] & \subset V_{0}, \label{V000} \\
\left[ V_{0},V_{1}\right] & \subset V_{1}, \\
\left[ V_{1},V_{1}\right] & \subset V_{0}.  \label{dos14}
\end{align}

The second step consists of finding an abelian semigroup $S$ which can be
partitioned in a ``resonant'' way with
respect to eqs.~(\ref{V000})--(\ref{dos14}).
We shall consider the expansion procedure using two different semigroups.

\subsection{Semigroup $S_{\text{S3}}$}
\label{sec:S3}

Let us consider first the semigroup
$S_{\text{S3}} = \left\{
\bar{\lambda}_{0},
\bar{\lambda}_{1},
\bar{\lambda}_{2},
\bar{\lambda}_{3}
\right\}$
defined by the following multiplication table: 
\begin{equation}
\begin{tabular}{c|cccc}
                    & $\bar{\lambda}_{0}$ & $\bar{\lambda}_{1}$ & $\bar{\lambda}_{2}$ & $\bar{\lambda}_{3}$ \\
\hline
$\bar{\lambda}_{0}$ & $\bar{\lambda}_{2}$ & $\bar{\lambda}_{3}$ & $\bar{\lambda}_{0}$ & $\bar{\lambda}_{3}$ \\ 
$\bar{\lambda}_{1}$ & $\bar{\lambda}_{3}$ & $\bar{\lambda}_{1}$ & $\bar{\lambda}_{3}$ & $\bar{\lambda}_{3}$ \\ 
$\bar{\lambda}_{2}$ & $\bar{\lambda}_{0}$ & $\bar{\lambda}_{3}$ & $\bar{\lambda}_{2}$ & $\bar{\lambda}_{3}$ \\ 
$\bar{\lambda}_{3}$ & $\bar{\lambda}_{3}$ & $\bar{\lambda}_{3}$ & $\bar{\lambda}_{3}$ & $\bar{\lambda}_{3}$
\end{tabular}
\label{dos15}
\end{equation}
A straightforward but important observation is that, for each $\lambda_{\alpha }\in S$, we have that $\bar{\lambda}_{3}\bar{\lambda}_{\alpha }=\bar{\lambda}_{3}$,
so that $\bar{\lambda}_{3}$ is seen to play the r\^{o}le of the zero element inside $S$.

Consider now the partition $S=S_{0}\cup S_{1}$, with
\begin{align}
S_{0}& =\left\{ \bar{\lambda}_{1},\bar{\lambda}_{2},\bar{\lambda}%
_{3}\right\} ,  \label{dos10} \\
S_{1}& =\left\{ \bar{\lambda}_{0},\bar{\lambda}_{3}\right\}.
\label{dos11}
\end{align}%
This partition is said to be resonant, since it satisfies
[cf.~eqs.~(\ref{V000})--(\ref{dos14})]
\begin{align}
S_{0} \cdot S_{0} & \subset S_{0}, \\
S_{0} \cdot S_{1} & \subset S_{1}, \\
S_{1} \cdot S_{1} & \subset S_{0}. \label{dos12}
\end{align}

Theorem~4.2 from Ref.~\cite{Iza06b} now assures us that 
\begin{equation}
\mathfrak{G}_{\text{R}}=W_{0}\oplus W_{1}  \label{dos13}
\end{equation}%
is a \emph{resonant subalgebra} of $S_{\text{S3}}\times \mathfrak{g,}$where%
\begin{align}
W_{0} & =
S_{0} \times V_{0} =
\left\{
  \bar{\lambda}_{1}, \bar{\lambda}_{2},\bar{\lambda}_{3}
\right\}
\otimes
\left\{ \bar{\bm{J}}_{ab} \right\} =
\left\{
  \bar{\lambda}_{1} \bar{\bm{J}}_{ab},
  \bar{\lambda}_{2} \bar{\bm{J}}_{ab},
  \bar{\lambda}_{3} \bar{\bm{J}}_{ab}
\right\},
\label{2.14} \\
W_{1} & =
S_{1} \times V_{1} =
\left\{
  \bar{\lambda}_{0},
  \bar{\lambda}_{3}
\right\}
\otimes
\left\{ \bar{\bm{P}}_{a} \right\} =
\left\{
  \bar{\lambda}_{0} \bar{\bm{P}}_{a},
  \bar{\lambda}_{3} \bar{\bm{P}}_{a}
\right\}.
\label{2.15}
\end{align}

As a last step, impose the condition $\lambda _{3}\times \mathfrak{g}=0$ on
$\mathfrak{G}_{\text{R}}$ and relabel its generators as
$\bm{J}_{ab,1} = \bar{\lambda}_{1} \bar{\bm{J}}_{ab}$,
$\bm{J}_{ab,2} = \bar{\lambda}_{2} \bar{\bm{J}}_{ab}$, and
$\bm{P}_{a,0}  = \bar{\lambda}_{0} \bar{\bm{P}}_{a}$.
This procedure leads us to the following commutation relations:
\begin{align}
\left[ \bm{J}_{ab,1}, \bm{J}_{cd,1}\right] & =
\bar{\lambda}_{1} \bar{\lambda}_{1}
\left[ \bar{\bm{J}}_{ab}, \bar{\bm{J}}_{cd} \right]
\nonumber \\ & =
\bar{\lambda}_{1}
\left[ \bar{\bm{J}}_{ab}, \bar{\bm{J}}_{cd}\right]
\nonumber \\ & =
\eta_{ad} \bm{J}_{bc,1} + \eta_{bc} \bm{J}_{ad,1} - \eta_{ac} \bm{J}_{bd,1} - \eta_{bd} \bm{J}_{ac,1},
\label{2.16a}
\end{align}
\begin{align}
\left[ \bm{J}_{ab,2}, \bm{J}_{cd,2} \right] & =
\bar{\lambda}_{2} \bar{\lambda}_{2}
\left[ \bar{\bm{J}}_{ab}, \bar{\bm{J}}_{cd} \right]
\nonumber \\ & =
\bar{\lambda}_{2}
\left[ \bar{\bm{J}}_{ab}, \bar{\bm{J}}_{cd} \right]
\nonumber \\ & =
\eta_{ad} \bm{J}_{bc,2} + \eta_{bc} \bm{J}_{ad,2} - \eta_{ac} \bm{J}_{bd,2} - \eta_{bd} \bm{J}_{ac,2},
\label{2.16b}
\end{align}
\begin{align}
\left[ \bm{J}_{ab,1}, \bm{J}_{cd,2} \right] & =
\bar{\lambda}_{1} \bar{\lambda}_{2}
\left[ \bar{\bm{J}}_{ab}, \bar{\bm{J}}_{cd} \right]
\nonumber \\ & =
\bar{\lambda}_{3}
\left[ \bar{\bm{J}}_{ab}, \bar{\bm{J}}_{cd} \right]
\nonumber \\ & = 0,
\label{2.16c}
\end{align}
\begin{align}
\left[ \bm{J}_{ab,1}, \bm{P}_{c,0} \right] & =
\bar{\lambda}_{1} \bar{\lambda}_{0}
\left[ \bar{\bm{J}}_{ab}, \bar{\bm{P}}_{c} \right]
\nonumber \\ & =
\bar{\lambda}_{3}
\left[ \bar{\bm{J}}_{ab}, \bar{\bm{P}}_{c}\right]
\nonumber \\ & = 0,
\label{2.16d}
\end{align}
\begin{align}
\left[ \bm{J}_{ab,2}, \bm{P}_{c,0} \right] & =
\bar{\lambda}_{2} \bar{\lambda}_{0}
\left[ \bar{\bm{J}}_{ab}, \bar{\bm{P}}_{c} \right]
\nonumber \\ & =
\bar{\lambda}_{0} \left[ \bar{\bm{J}}_{ab}, \bar{\bm{P}}_{c} \right]
\nonumber \\ & =
\eta _{bc} \bm{P}_{a,0} - \eta _{ac} \bm{P}_{b,0}
\label{2.16e}
\end{align}
\begin{align}
\left[ \bm{P}_{a,0}, \bm{P}_{b,0} \right] & =
\bar{\lambda}_{0} \bar{\lambda}_{0}
\left[ \bar{\bm{P}}_{a}, \bar{\bm{P}}_{b} \right]
\nonumber \\ & =
\bar{\lambda}_{2}
\left[ \bar{\bm{P}}_{a}, \bar{\bm{P}}_{b} \right]
\nonumber \\ & =
\bm{J}_{ab,2},
\label{2.16f}
\end{align}
where we have used the commutation relations of the AdS algebra
and the multiplication law~(\ref{dos15}) of the semigroup $S_{S3}$.

The identification
$\bm{N}_{ab} = \bm{J}_{ab,1}$,
$\bm{L}_{ab} = \bm{J}_{ab,2}$, and
$\bm{L}_{a,D} = \bm{P}_{a,0}$
shows that the algebra~(\ref{2.16a})--(\ref{2.16f}), obtained by $S_{S3}$-expansion and $0_{S}$-reduction of
the AdS algebra $\mathfrak{so} \left( D-1,2\right)$, coincides with the SSEP algebra~(\ref{2.12'})--(\ref{2.13})
obtained by semisimple extension of the Poincar\'{e} algebra in Refs.~\cite{Sor04,Dup05,Sor06}.

\subsection{Semigroup $S_{\text{S2}}$}
\label{sec:S2}

Let us now consider the semigroup
$S_{\text{S2}}=
\left\{
  \lambda_{0},
  \lambda_{1},
  \lambda_{2}
\right\}$
defined by the multiplication law
\begin{equation}
\lambda_{\alpha} \lambda_{\beta} =
\left\{
\begin{array}{ll}
\lambda_{\alpha + \beta}   & \text{if } \alpha + \beta \leq 2 \\ 
\lambda_{\alpha +\beta -2} & \text{if } \alpha + \beta > 2
\end{array}
\right.,
\label{3.1}
\end{equation}
or, equivalently, by the multiplication table
\begin{equation}
\begin{array}{c|ccc}
& \lambda_{0} & \lambda_{1} & \lambda_{2} \\
\hline
\lambda_{0} & \lambda_{0} & \lambda_{1} & \lambda _{2} \\ 
\lambda_{1} & \lambda_{1} & \lambda_{2} & \lambda _{1} \\ 
\lambda_{2} & \lambda_{2} & \lambda_{1} & \lambda _{2}
\end{array}
\label{3.2}
\end{equation}

Take now the partition $S=S_{0}\cup S_{1}$, with
\begin{align}
S_{0} & =\left\{ \lambda _{0},\lambda _{2}\right\} ,  \label{3.3} \\
S_{1} & =\left\{ \lambda _{1} \right\} .  \label{3.4}
\end{align}
This partition is said to be resonant, since it satisfies 
[cf.~eqs.~(\ref{V000})--(\ref{dos14})]
\begin{align}
S_{0} \cdot S_{0} & \subset S_{0}, \\
S_{0} \cdot S_{1} & \subset S_{1}, \\
S_{1} \cdot S_{1} & \subset S_{0}.
\end{align}

Theorem~4.2 from Ref.~\cite{Iza06b} now assures us that 
\begin{equation}
\mathfrak{G}_{\text{R}} = W_{0} \oplus W_{1}
\end{equation}
is a \emph{resonant subalgebra} of $S_{\text{S2}}\times \mathfrak{g}$, where
\begin{align}
W_{0} & =
S_{0} \times V_{0} =
\left\{
  \lambda_{0},
  \lambda_{2}
\right\}
\otimes
\left\{
  \bar{\bm{J}}_{ab}
\right\} =
\left\{ 
  \lambda_{0} \bar{\bm{J}}_{ab},
  \lambda_{2} \bar{\bm{J}}_{ab}
\right\},
\\
W_{1} & =
S_{1} \times V_{1} =
\left\{
  \lambda_{1}
\right\}
\otimes
\left\{
  \bar{\bm{P}}_{a}
\right\} =
\left\{
  \lambda_{1} \bar{\bm{P}}_{a}
\right\}.
\end{align}

Relabeling the generators of the resonant subalgebra as
$\bar{\bm{J}}_{ab,0} = \lambda_{0} \bar{\bm{J}}_{ab}$,
$\bar{\bm{J}}_{ab,2} = \lambda_{2} \bar{\bm{J}}_{ab}$, and
$\bar{\bm{P}}_{a,1} = \lambda_{1} \bar{\bm{P}}_{a}$,
we are left with the following commutation relations:
\begin{align}
\left[
  \bar{\bm{J}}_{ab,0},
  \bar{\bm{J}}_{cd,0}
\right]
& =
\lambda_{0} \lambda _{0}
\left[
  \bar{\bm{J}}_{ab},
  \bar{\bm{J}}_{cd}
\right]
\nonumber \\ & =
\lambda_{0}
\left[
  \bar{\bm{J}}_{ab},
  \bar{\bm{J}}_{cd}
\right]
\nonumber \\ & =
\eta_{ad} \bar{\bm{J}}_{bc,0} +
\eta_{bc} \bar{\bm{J}}_{ad,0} -
\eta_{ac} \bar{\bm{J}}_{bd,0} -
\eta_{bd} \bar{\bm{J}}_{ac,0},
\end{align}
\begin{align}
\left[
  \bar{\bm{J}}_{ab,2},
  \bar{\bm{J}}_{cd,2}
\right]
& =
\lambda_{2} \lambda _{2}
\left[
  \bar{\bm{J}}_{ab},
  \bar{\bm{J}}_{cd}
\right]
\nonumber \\ & =
\lambda_{2}
\left[
  \bar{\bm{J}}_{ab},
  \bar{\bm{J}}_{cd}
\right]
\nonumber \\ & =
\eta_{ad} \bar{\bm{J}}_{bc,2} +
\eta_{bc} \bar{\bm{J}}_{ad,2} -
\eta_{ac} \bar{\bm{J}}_{bd,2} -
\eta_{bd} \bar{\bm{J}}_{ac,2},
\end{align}
\begin{align}
\left[
  \bar{\bm{J}}_{ab,0},
  \bar{\bm{J}}_{cd,2}
\right]
& =
\lambda_{0} \lambda_{2}
\left[
  \bar{\bm{J}}_{ab},
  \bar{\bm{J}}_{cd}
\right]
\nonumber \\ & =
\lambda_{2}
\left[
  \bar{\bm{J}}_{ab},
  \bar{\bm{J}}_{cd}
\right]
\nonumber \\ & =
\eta_{ad} \bar{\bm{J}}_{bc,2} +
\eta_{bc} \bar{\bm{J}}_{ad,2} -
\eta_{ac} \bar{\bm{J}}_{bd,2} -
\eta_{bd} \bar{\bm{J}}_{ac,2},
\end{align}
\begin{align}
\left[
  \bar{\bm{J}}_{ab,0},
  \bar{\bm{P}}_{c,1}
\right]
& =
\lambda_{0} \lambda _{1}
\left[ 
  \bar{\bm{J}}_{ab},
  \bar{\bm{P}}_{c}
\right]
\nonumber \\ & =
\lambda_{1}
\left[
  \bar{\bm{J}}_{ab},
  \bar{\bm{P}}_{c}
\right]
\nonumber \\ & =
\eta_{cb} \bar{\bm{P}}_{a,1} -
\eta_{ac} \bar{\bm{P}}_{b,1},
\end{align}
\begin{align}
\left[
  \bar{\bm{J}}_{ab,2},
  \bar{\bm{P}}_{c,1}
\right]
& =
\lambda_{2} \lambda _{1}
\left[ 
  \bar{\bm{J}}_{ab},
  \bar{\bm{P}}_{c}
\right]
\nonumber \\ & =
\lambda_{1}
\left[
  \bar{\bm{J}}_{ab},
  \bar{\bm{P}}_{c}
\right]
\nonumber \\ & =
\eta_{bc} \bar{\bm{P}}_{a,1} -
\eta_{ac} \bar{\bm{P}}_{b,1},
\end{align}
\begin{align}
\left[
  \bar{\bm{P}}_{a,1},
  \bar{\bm{P}}_{b,1}
\right]
& =
\lambda_{1} \lambda_{1}
\left[ 
  \bar{\bm{P}}_{a},
  \bar{\bm{P}}_{b}
\right]
\nonumber \\ & =
\lambda_{2}
\left[
  \bar{\bm{P}}_{a},
  \bar{\bm{P}}_{b}
\right]
\nonumber \\ & =
\bar{\bm{J}}_{ab,2},
\end{align}
where we have used the commutation relations of the AdS algebra
and the multiplication law~(\ref{3.1}) of the semigroup $S_{\text{S2}}$.

The identifications
$\tilde{\bm{J}}_{ab} = \bar{\bm{J}}_{ab,0}$,
$\tilde{\bm{Z}}_{ab} = \bar{\bm{J}}_{ab,2}$, and
$\tilde{\bm{P}}_{a} = \bar{\bm{P}}_{a,1}$
lead to the following algebra:
\begin{align}
\left[
  \tilde{\bm{J}}_{ab},
  \tilde{\bm{J}}_{cd}
\right]
& =
\eta_{ad} \tilde{\bm{J}}_{bc} +
\eta_{bc} \tilde{\bm{J}}_{ad} -
\eta_{ac} \tilde{\bm{J}}_{bd} -
\eta_{bd} \tilde{\bm{J}}_{ac},
\\
\left[
  \tilde{\bm{J}}_{ab},
  \tilde{\bm{P}}_{c}
\right]
& =
\eta_{bc} \tilde{\bm{P}}_{a} -
\eta_{ac} \tilde{\bm{P}}_{b},
\\
\left[
  \tilde{\bm{P}}_{a},
  \tilde{\bm{P}}_{b}
\right]
& =
\tilde{\bm{Z}}_{ab},
\\
\left[
  \tilde{\bm{J}}_{ab},
  \tilde{\bm{Z}}_{cd}
\right]
& =
\eta_{ad} \tilde{\bm{Z}}_{bc} +
\eta_{bc} \tilde{\bm{Z}}_{ad} -
\eta_{ac} \tilde{\bm{Z}}_{bd} -
\eta_{bd} \tilde{\bm{Z}}_{ac},
\\
\left[
  \tilde{\bm{Z}}_{ab},
  \tilde{\bm{P}}_{c}
\right]
& =
\eta_{bc} \tilde{\bm{P}}_{a} -
\eta_{ac} \tilde{\bm{P}}_{b},
\\
\left[
  \tilde{\bm{Z}}_{ab},
  \tilde{\bm{Z}}_{cd}
\right]
& =
\eta_{ad} \tilde{\bm{Z}}_{bc} +
\eta_{bc} \tilde{\bm{Z}}_{ad} -
\eta_{ac} \tilde{\bm{Z}}_{bd} -
\eta_{bd} \tilde{\bm{Z}}_{ac},
\end{align}
which matches the SSEP algebra~(\ref{eext1})--(\ref{2.1}) obtained in Refs.~\cite{Sor04,Dup05,Sor06,Sor10}, up to (inessential) numerical factors.

\subsection{Relationship between the multiplication tables of the semigroups
$S_{\text{S3}}$ and $S_{\text{S2}}$}

In section~\ref{sec:S3}, the SSEP algebra~(\ref{2.12'})--(\ref{2.13}) was obtained
through an $S$-expansion using the semigroup $S_{\text{S3}}$,
whose multiplication table is given in eq.~(\ref{dos15}).
The procedure involves imposing the condition known as $0_{S}$-reduction~\cite{Iza06b}.

In section~\ref{sec:S2}, the SSEP algebra~(\ref{eext1})--(\ref{2.1}) was obtained
(up to inessential numerical factors)
through an $S$-expansion using the semigroup $S_{\text{S2}}$,
whose multiplication table is given in eq.~(\ref{3.2}).
In stark contrast with the previous case, the procedure does not involve imposing the $0_{S}$-reduction.

This curious state of affairs can be clarified by
promoting the semigroup $S_{\text{S2}}$ to a \emph{ring}%
\footnote{Here we do not require that the elements of the ring form a group under multiplication, but rather only a semigroup.}
and setting
\begin{align}
\tilde{\lambda}_{1} & = \lambda _{0} - \lambda _{2}, \\
\tilde{\lambda}_{2} & = \lambda _{2},  \label{3.5} \\
\tilde{\lambda}_{0} & = \lambda _{1}.
\end{align}
This amounts to a change of basis in $S_{\text{S2}}$ and leads to the following multiplicacion table:
\begin{equation}
\begin{array}{c|ccc}
& \tilde{\lambda}_{0} & \tilde{\lambda}_{1} & \tilde{\lambda}_{2} \\
\hline
\tilde{\lambda}_{0} & \tilde{\lambda}_{2} & 0 & \tilde{\lambda}_{0} \\ 
\tilde{\lambda}_{1} & 0 & \tilde{\lambda}_{1} & 0 \\ 
\tilde{\lambda}_{2} & \tilde{\lambda}_{0} & 0 & \tilde{\lambda}_{2}
\end{array}
\label{3.6}
\end{equation}
This multiplication table exactly matches the multiplication table of the $S_{\text{S3}}$ semigroup
[see eq.~(\ref{dos15})],
except for the rows and columns involving $\lambda_{3}$.
In place of $\lambda_{3}$, the symbol ``0'' in~(\ref{3.6}) now stands for the \emph{additive} zero of the $S_{\text{S2}}$ ring.

The generators $\bm{N}_{ab}$ and $\bm{L}_{AB}$ can be recovered by setting
\begin{align}
\bm{N}_{ab} & =
\tilde{\lambda}_{1} \bar{\bm{J}}_{ab} =
\left(
  \lambda_{0} - \lambda_{2}
\right)
\bar{\bm{J}}_{ab},
\\
\bm{L}_{ab} & =
\tilde{\lambda}_{2} \bar{\bm{J}}_{ab} =
\lambda_{2} \bar{\bm{J}}_{ab},
\\
\bm{L}_{a,D} & =
\tilde{\lambda}_{0} \bar{\bm{P}}_{a} =
\lambda_{1} \bar{\bm{P}}_{a},
\end{align}
without invoking the $0_{S}$-reduction.
The advantage of not using the $0_{S}$-reduction is that it facilitates the construction of Casimir operators,
as discussed in section~\ref{sec:Casgral}.

\section{Casimir operators for $S$-expanded Lie Algebras}
\label{sec:Casgral}

In this section we consider the construction of Casimir operators for
$S$-expanded Lie algebras.
We then compute the Casimir operators for the SSEP algebra
obtained by Soroka et~al.\ in Refs.~\cite{Sor04,Dup05,Sor06,Sor10}.

\subsection{Construction of Casimir operators for $S$-expanded Lie algebras}

Let $\mathfrak{g}$ be a Lie algebra and let
$\left\{ \bm{T}_{A}, A=1,\ldots,\dim \mathfrak{g} \right\}$
be the generators of $\mathfrak{g}$.
A Casimir operator $\bm{C}_{m}$ of degree $m$ can be written as 
\begin{equation}
\bm{C}_{m} =
C^{A_{1 }\cdots A_{m}}
\bm{T}_{A_{1}} \cdots \bm{T}_{A_{m}},
\label{3.7}
\end{equation}
which, by definition, satisfies the condition that,
$\forall~\bm{T}_{A_{0}} \in \mathfrak{g}$,
\begin{equation}
\left[
  \bm{T}_{A_{0}},
  \bm{C}_{m}
\right]
= 0,
\label{3.8}
\end{equation}
where the coefficients $C^{A_{1} \cdots A_{m}}$ form a symmetric invariant tensor for the corresponding Lie group.
This means that the operators $\bm{C}_{m}$ ($m=2,3,\ldots$) are invariants of the enveloping algebra.
From eqs.~(\ref{3.7}) and~(\ref{3.8}) we have
\begin{equation}
\left[
  \bm{T}_{A_{0}},
  \bm{C}_{m}
\right] =
\left(
  \sum_{p=1}^{m}
  \swne{f}{A_{0} B}{A_{p}}
  C^{A_{1} \cdots A_{p-1} B A_{p+1} \cdots A_{m}}
\right)
\bm{T}_{A_{1}}
\cdots
\bm{T}_{A_{m}},
\label{3.9}
\end{equation}
where $\swne{f}{AB}{C}$ are the structure constants of $\mathfrak{g}$.
Therefore, the ``Casimir Condition''~(\ref{3.8}) is seen to be equivalent to
\begin{equation}
\sum_{p=1}^{m}
\swne{f}{A_{0} B}{(A_{p}}
C^{A_{1} \cdots A_{p-1} | B | A_{p+1} \cdots A_{m})} = 0.
\label{3.10}
\end{equation}

For the standard, quadratic (i.e., $m=2$) Casimir operator, eq.~(\ref{3.10}) reads
\begin{equation}
\swne{f}{A_{0} B}{A_{1}}
C^{B A_{2}} +
\swne{f}{A_{0} B}{A_{2}}
C^{A_{1} B} = 0.
\label{3.11}
\end{equation}

The structure constants of an $S$-expanded Lie algebra are given by
\begin{equation}
\swne{f}{\left( A, \alpha \right) \left( B, \beta \right)}{\left( C, \gamma \right)} =
\swne{K}{\alpha \beta}{\gamma}
\swne{f}{AB}{C},
\label{Ec_S_Ctes}
\end{equation}
where $\swne{K}{\alpha \beta}{\gamma}$ stands for the ``two-selector'' of the semigroup $S$~\cite{Iza06b}.
The (quadratic) Casimir condition for an $S$-expanded Lie algebra thus reads
\begin{equation}
\swne{K}{\alpha_{0} \beta}{\alpha_{1}}
\swne{f}{A_{0} B}{A_{1}}
C^{\left( B, \beta \right) \left( A_{2}, \alpha _{2} \right)} +
\swne{K}{\alpha_{0} \beta}{\alpha_{2}}
\swne{f}{A_{0} B}{A_{2}}
C^{\left( A_{1}, \alpha_{1} \right) \left( B, \beta \right)} = 0.
\label{Ec_Casimir_Cuadratico}
\end{equation}

Consider now the following ansatz for the components of the (quadratic) Casimir operator of an $S$-expanded algebra:
\begin{equation}
C^{\left( A, \alpha \right) \left( B, \beta \right)} =
m^{\alpha \beta } C^{AB},
\label{3.12}
\end{equation}%
where $C^{AB}$ are the components of the (quadratic) Casimir operator for the original
algebra $\mathfrak{g}$ and $m^{\alpha \beta}$ are the components of a
symmetric tensor, associated to the semigroup $S$, which must be determined.

Introducing~(\ref{Ec_S_Ctes}) in~(\ref{Ec_Casimir_Cuadratico}) we obtain%
\begin{equation}
\swne{K}{\alpha_{0} \beta}{\alpha_{1}}
m^{\beta \alpha_{2}}
\swne{f}{A_{0} B}{A_{1}}
C^{B A_{2}} +
\swne{K}{\alpha_{0} \beta}{\alpha_{2}}
m^{\alpha_{1} \beta}
\swne{f}{A_{0} B}{A_{2}}
C^{A_{1} B} = 0.
\label{3.13}
\end{equation}
Eq.~(\ref{3.13}) is satisfied if the following condition holds:
\begin{equation}
\swne{K}{\alpha_{0} \beta}{\alpha_{1}}
m^{\beta \alpha_{2}} =
\swne{K}{\alpha_{0} \beta}{\alpha_{2}}
m^{\alpha_{1} \beta}.
\label{Ec_Cond_m}
\end{equation}
To check this, let us plug eq.~(\ref{Ec_Cond_m}) into eq.~(\ref{3.13}) to find
\begin{equation}
\swne{K}{\alpha_{0} \beta}{\alpha_{1}}
m^{\beta \alpha_{2}}
\swne{f}{A_{0} B}{A_{1}}
C^{B A_{2}} +
\swne{K}{\alpha_{0} \beta}{\alpha_{2}}
m^{\alpha_{1} \beta}
\swne{f}{A_{0} B}{A_{2}}
C^{A_{1} B} =
\swne{K}{\alpha_{0} \beta}{\alpha_{1}}
m^{\beta \alpha_{2}}
\left(
  \swne{f}{A_{0} B}{A_{1}}
  C^{B A_{2}} +
  \swne{f}{A_{0} B}{A_{2}}
  C^{A_{1} B}
\right) = 0,
\end{equation}
where the expression in parentheses vanishes because $C^{AB}$
are the components of the (quadratic) Casimir operator for the
original algebra $\mathfrak{g}$
[cf.~eq.~(\ref{3.11})].

The following Theorem provides us with a way of finding a tensor $m^{\alpha \beta}$
with the required properties.

\begin{theorem}
Let $\swne{K}{\alpha \beta}{\gamma}$ be the two-selector for an abelian semigroup $S$, and define
\begin{equation}
m_{\alpha \beta} = \alpha_{\gamma} \swne{K}{\alpha \beta}{\gamma},
\label{3.15}
\end{equation}
where the $\alpha_{\gamma}$ are numerical coefficients.
If the $\alpha_{\gamma}$ are chosen in such a way that $m_{\alpha \beta}$ is an invertible ``metric,''
then its inverse
$m^{\alpha \beta}$
(which, by definition, satisfies
$m^{\alpha \lambda} m_{\lambda \beta} = \delta_{\beta}^{\alpha}$)
fulfills eq.~(\ref{Ec_Cond_m}).
\end{theorem}

\begin{proof}
From the associativity and commutativity of the inner binary
operation (``multiplication'') of the semigroup $S$, we have
\begin{equation}
\left(
  \lambda_{\alpha_{0}}
  \lambda_{\mu}
\right)
\lambda_{\nu} =
\left(
  \lambda_{\alpha_{0}}
  \lambda_{\nu}
\right)
\lambda_{\mu}.
\label{3.16}
\end{equation}%
In terms of the two-selectors $\swne{K}{\alpha \beta}{\gamma}$,
eq.~(\ref{3.16}) may be cast as
\begin{equation}
\swne{K}{\alpha_{0} \mu}{\alpha_{1}}
\swne{K}{\alpha_{1} \nu}{\lambda} =
\swne{K}{\alpha_{0} \nu}{\alpha_{2}}
\swne{K}{\alpha_{2} \mu}{\lambda}.
\label{3.17}
\end{equation}
Multiplying~(\ref{3.17}) by $\alpha_{\lambda}$, we find 
\begin{align}
\swne{K}{\alpha_{0} \mu}{\alpha_{1}}
m_{\alpha_{1} \nu}
& =
\swne{K}{\alpha_{0} \nu}{\alpha_{2}}
m_{\alpha_{2} \mu}
\nonumber \\
\swne{K}{\alpha_{0} \beta}{\alpha_{1}}
\delta_{\mu}^{\beta}
m_{\alpha_{1} \nu}
& =
\swne{K}{\alpha_{0} \beta}{\alpha_{2}}
\delta_{\nu}^{\beta}
m_{\alpha_{2} \mu}
\nonumber \\
\swne{K}{\alpha_{0} \beta}{\alpha_{1}}
m^{\beta \alpha_{2}}
m_{\alpha_{2} \mu}
m_{\alpha_{1} \nu}
& =
\swne{K}{\alpha_{0} \beta}{\alpha_{2}}
m^{\alpha_{1} \beta}
m_{\alpha_{2} \mu}
m_{\alpha_{1} \nu},
\label{3.18}
\end{align}
so that
\begin{equation}
\swne{K}{\alpha_{0} \beta}{\alpha_{1}}
m^{\beta \alpha_{2}} =
\swne{K}{\alpha_{0} \beta}{\alpha_{2}}
m^{\alpha_{1} \beta},
\label{3.19}
\end{equation}
as we set out to prove.
\end{proof}

This means that if
$\bm{C} = C^{AB} \bm{T}_{A} \bm{T}_{B}$
is the (quadratic) Casimir operator for the original algebra $\mathfrak{g}$, then
\begin{equation}
\bm{C} =
m^{\alpha \beta}
C^{AB}
\bm{T}_{\left( A, \alpha \right)}
\bm{T}_{\left( B, \beta \right)}
\label{3.20}
\end{equation}
is the (quadratic) Casimir operator for the $S$-expanded Lie algebra.

\subsection{Casimir operators for anti-de~Sitter algebra}

Using the representation given by the Dirac matrices for the AdS algebra,
\begin{align}
\bm{P}_{a} & =
\frac{1}{2} \Gamma_{a}, \\
\bm{J}_{ab} & =
\frac{1}{2} \Gamma_{ab},
\label{3.21}
\end{align}
we have that the Killing metric $k_{AB}$ for the AdS algebra can be written as 
\begin{align}
k_{AB} & =
\frac{1}{\Tr \left( \openone \right)}
\Tr
\left(
  \bm{T}_{A}
  \bm{T}_{B}
\right)
\\ & =
\frac{1}{\Tr \left( \openone \right)}
\Tr
\left(
  \frac{1}{2}
  \left\{
    \bm{T}_{A},
    \bm{T}_{B}
  \right\}
\right),
\end{align}
which for $d\geq 4$ is given by
\begin{align}
k_{a,b} & = \frac{1}{4} \eta_{ab}  \\
k_{ab,cd} & = -\frac{1}{4} \eta_{\left[ ab \right] \left[ cd \right]},
\label{3.22} \\
k_{ab,c} & = 0,
\label{eq:m12=0}
\end{align}
where
\begin{equation}
\eta_{\left[ ab \right] \left[ cd \right]} =
\delta_{ab}^{mn} \eta_{mc} \eta_{nd}.
\label{3.23}
\end{equation}

For an arbitrary algebra, the quadratic Casimir operator is given by 
\begin{equation}
\bm{C} = k^{AB} \bm{T}_{A} \bm{T}_{B},
\label{3.24}
\end{equation}
where $k^{AB}$ stands for the inverse of the Killing metric $k_{AB}$.

For the AdS algebra we have 
\begin{align}
k^{a,b} & = 4 \eta^{ab} \label{eq:kAdSab}\\
k^{ab,c} & = 0, \\
k^{ab,cd} & = - \eta^{\left[ ab \right] \left[ cd \right]},
\label{3.25}
\end{align}%
so that 
\begin{equation}
\bm{C}_{\text{AdS}} = 4
\left(
  \bm{P}^{a}
  \bm{P}_{a}
  - \frac{1}{2}
  \bm{J}_{ab}
  \bm{J}^{ab}
\right).
\label{3.26}
\end{equation}
This result is valid for any dimension $d \geq 4$.

There is another Killing ``metric'' that can be constructed only in $d=4$.
This is given by 
\begin{equation}
\bar{k}_{AB} =
\frac{1}{\Tr \left( \openone \right)}
\Tr
\left(
  \Gamma_{\ast}
  \bm{T}_{A}
  \bm{T}_{B}
\right),
\label{3.27}
\end{equation}%
where $\Gamma _{\ast}$ is the usual $\gamma_{5}$ matrix.
A direct calculation shows that
\begin{align}
\bar{k}_{a,b} & = 0, \\
\bar{k}_{ab,cd} & = - \frac{1}{4} \epsilon_{abcd}, \\
\bar{k}_{ab,c} & = 0.
\label{3.28}
\end{align}
This ``metric,'' however, is not invertible,
so that we cannot construct a Casimir operator for the AdS algebra from it.
On second thought, it \emph{is} possible to use this ``metric'' to construct a Casimir operator for the
Lorentz subalgebra, because, when so restricted, the metric turns out to be invertible.
We find
\begin{equation}
\bar{k}^{ab,cd} = - \epsilon^{abcd}.
\label{3.29}
\end{equation}
This means that a (quadratic) Casimir operator for the Lorentz group is given by 
\begin{equation}
\bm{\bar{C}}_{\text{L}} = - \epsilon^{abcd} \bm{J}_{ab} \bm{J}_{cd}.
\label{3.30}
\end{equation}

\section{Casimir Operators for the semi-simple extended Poincar\'{e} algebra}
\label{sec:Caspart}

We consider the construction of the metric $m_{\alpha \beta }$ corresponding
to the semigroup $S_{\text{S2}}$, whose multiplication law is given in eqs.~(\ref{3.1}) and~(\ref{3.2}).
The semigroup $S_{\text{S2}}$ is interesting
because, although it is not a group, it is \emph{cyclic} (i.e., similar to $\mathbbm{Z}_{3}$). The
elements of the semigroup can be represented by the following set of matrices:
\begin{equation}
\lambda_{0} =
\left( 
\begin{array}{ccc}
1 & 0 & 0 \\ 
0 & 1 & 0 \\ 
0 & 0 & 1%
\end{array}%
\right),
\quad
\lambda_{1} =
\left( 
\begin{array}{ccc}
0 & 0 & 0 \\ 
1 & 0 & 1 \\ 
0 & 1 & 0%
\end{array}%
\right),
\quad
\lambda_{2} =
\left( 
\begin{array}{ccc}
0 & 0 & 0 \\ 
0 & 1 & 0 \\ 
1 & 0 & 1%
\end{array}%
\right).
\label{3.31}
\end{equation}%
It is straightforward to verify that the representation~(\ref{3.31}) faithfully
satisfies eqs.~(\ref{3.1}) and~(\ref{3.2}).
The two-selectors
$\swne{K}{\alpha \beta}{\gamma}$
of $S_{\text{S2}}$ can be represented as
[cf.~eqs.~(1)--(2) from Ref.~\cite{Iza06b}]
\begin{equation}
\swne{K}{\alpha \beta}{0} =
\left( 
\begin{array}{ccc}
1 & 0 & 0 \\ 
0 & 0 & 0 \\ 
0 & 0 & 0%
\end{array}%
\right),
\quad
\swne{K}{\alpha \beta}{1} =
\left( 
\begin{array}{ccc}
0 & 1 & 0 \\ 
1 & 0 & 1 \\ 
0 & 1 & 0%
\end{array}%
\right),
\quad
\swne{K}{\alpha \beta}{2} =
\left( 
\begin{array}{ccc}
0 & 0 & 1 \\ 
0 & 1 & 0 \\ 
1 & 0 & 1%
\end{array}%
\right).
\label{3.32}
\end{equation}
This, in turn, implies that a generic metric $m_{\alpha \beta}$ for $S_{\text{S2}}$ reads
\begin{equation}
m_{\alpha \beta} =
\alpha_{\lambda}
\swne{K}{\alpha \beta}{\lambda} =
\left( 
\begin{array}{ccc}
\alpha_{0} & \alpha_{1} & \alpha_{2} \\ 
\alpha_{1} & \alpha_{2} & \alpha_{1} \\ 
\alpha_{2} & \alpha_{1} & \alpha_{2}
\end{array}
\right),
\label{3.33}
\end{equation}%
where the $\alpha_{\lambda}$ are numerical coefficients.
The inverse metric is given by
\begin{equation}
m^{\alpha \beta} =
\frac{1}{\det \left( m_{\alpha \beta} \right)}
\left( 
\begin{array}{ccc}
\alpha_{2}^{2} - \alpha_{1}^{2} &
0 &
-\left( \alpha_{2}^{2} - \alpha_{1}^{2} \right)
\\ 
0 &
\alpha_{2} \left( \alpha_{0} - \alpha_{2} \right) &
-\alpha_{1} \left( \alpha_{0} - \alpha_{2} \right)
\\ 
-\left( \alpha_{2}^{2} - \alpha_{1}^{2} \right) &
-\alpha_{1} \left( \alpha_{0} - \alpha_{2} \right) &
\alpha_{0} \alpha_{2} - \alpha_{1}^{2}
\end{array}%
\right),
\label{3.34}
\end{equation}%
where $\alpha_{0}$, $\alpha_{1}$ and $\alpha_{2}$ must be chosen so that
\begin{equation}
\det \left( m_{\alpha \beta} \right) =
\left( \alpha_{0} - \alpha_{2} \right)
\left( \alpha_{2} + \alpha_{1} \right)
\left( \alpha_{2} - \alpha_{1} \right)
\neq 0.
\label{3.35}
\end{equation}

The quadratic Casimir operators for the SSEP algebra has the form%
\footnote{The $m^{12}$-term is absent from the sum because the corresponding components of the Casimir operator for the AdS algebra in $d \geq 4$ vanish, $C^{ab,c} = C^{a,bc} =0$ [see eq.~(\ref{eq:m12=0})].}
\begin{align}
\bm{C} & =
m^{\alpha \beta}
C^{AB}
\bm{T}_{\left( A, \alpha \right)}
\bm{T}_{\left( B, \beta \right)}
\nonumber \\ & =
m^{00}
C^{ab,cd}
\bm{J}_{ab}
\bm{J}_{cd} +
m^{11}
C^{ab}
\bm{P}_{a}
\bm{P}_{b} +
2 m^{02}
C^{ab,cd}
\bm{J}_{ab}
\bm{Z}_{cd} +
m^{22}
C^{ab,cd}
\bm{Z}_{ab}
\bm{Z}_{cd}
\nonumber \\ & =
\frac{1}{\det \left( m_{\alpha \beta} \right)}
\left[
  \left( \alpha_{2}^{2} - \alpha_{1}^{2} \right)
  C^{ab,cd}
  \bm{J}_{ab}
  \bm{J}_{cd} +
  \alpha _{2} \left( \alpha_{0} - \alpha_{2} \right)
  C^{ab}
  \bm{P}_{a}
  \bm{P}_{b} +
\right.
  \nonumber \\ & 
\left. -
  2 \left( \alpha_{2}^{2} - \alpha_{1}^{2} \right)
  C^{ab,cd}
  \bm{J}_{ab}
  \bm{Z}_{cd} +
  \left( \alpha_{0} \alpha_{2} - \alpha_{1}^{2} \right)
  C^{ab,cd}
  \bm{Z}_{ab}
  \bm{Z}_{cd}
\right],
\label{3.36}
\end{align}
where $C^{AB}$ are the components of the Casimir operator for the AdS algebra.
Plugging eqs.~(\ref{eq:kAdSab})--(\ref{3.25}) into eq.~(\ref{3.36}) we find
\begin{equation}
\bm{C}=\frac{4}{\det \left( m_{\alpha \beta} \right)}\left[ \frac{1}{2}\left( \alpha
_{2}^{2}-\alpha _{1}^{2}\right) \bm{J}_{ab}\bm{J}%
^{ab}+\alpha _{2}\left( \alpha _{0}-\alpha _{2}\right) \bm{P}_{a}%
\bm{P}^{a}-\left( \alpha _{2}^{2}-\alpha _{1}^{2}\right) \bm{%
J}_{ab}\bm{Z}^{ab}+\frac{1}{2}\left( \alpha _{0}\alpha _{2}-\alpha
_{1}^{2}\right) \bm{Z}_{ab}\bm{Z}^{ab}\right] .  \label{3.37}
\end{equation}
Defining
\begin{align}
\alpha & =
\alpha_{2} \alpha_{0} - \alpha_{2}^{2}, \\
\beta & =
\alpha_{2} \alpha_{0} - \alpha_{1}^{2},
\label{3.38}
\end{align}%
eq.~(\ref{3.37}) can be cast in the form
\begin{align}
\bm{C}& =\frac{4}{\det \left( m_{\alpha \beta} \right)}\left[ \frac{1}{2}\left( \beta -\alpha
\right) \bm{J}_{ab}\bm{J}^{ab}+\alpha \bm{P}_{a}%
\bm{P}^{a}-\left( \beta -\alpha \right) \bm{J}_{ab}%
\bm{Z}^{ab}+\frac{1}{2}\beta \bm{Z}_{ab}\bm{Z}^{ab}%
\right]  \notag \\
& =\frac{4}{\det \left( m_{\alpha \beta} \right)}\left[ \alpha \left( \bm{P}_{a}\bm{P}^{a}-%
\frac{1}{2}\bm{J}_{ab}\bm{J}^{ab}+\bm{J}_{ab}%
\bm{Z}^{ab}\right) +\beta \left( \frac{1}{2}\bm{J}_{ab}%
\bm{J}^{ab}-\bm{J}_{ab}\bm{Z}^{ab}+\frac{1}{2}%
\bm{Z}_{ab}\bm{Z}^{ab}\right) \right]  \label{3.39}
\end{align}
Since $\alpha$ and $\beta$ are arbitrary,
subject only to the condition $\det \left( m_{\alpha \beta} \right) \neq 0$,
we can conclude that eq.~(\ref{3.39}) shows that
the SSEP possess
\emph{two} independent Casimir operators, namely
\begin{align}
\bm{C}_{1}& =\frac{4\alpha }{\det \left( m_{\alpha \beta} \right)}\left( \bm{P}_{a}%
\bm{P}^{a}+\bm{J}_{ab}\bm{Z}^{ab}-\frac{1}{2}%
\bm{J}_{ab}\bm{J}^{ab}\right) ,  \label{3.40} \\
\bm{C}_{2}& =\frac{2\beta }{\det \left( m_{\alpha \beta} \right)}\left( \bm{Z}_{ab}%
\bm{Z}^{ab}-2\bm{J}_{ab}\bm{Z}^{ab}+\bm{J}%
_{ab}\bm{J}^{ab}\right) .  \label{3.41}
\end{align}

There exists a third Casimir operator,
but it is valid only for the subspace spanned by $\bm{J}_{ab}$ and $\bm{Z}_{ab}$,
and not for the full SSEP algebra.
This Casimir operator is constructed from
$\bar{k}^{\left( ab,cd \right)} = - \epsilon^{abcd}$
[cf.~eq.~(\ref{3.29})],
and is given by
\begin{align}
\bar{\bm{C}}_{JZ} & =
-\frac{1}{\det \left( m_{\alpha \beta} \right)}
\left[
  \left( \alpha_{2}^{2} - \alpha_{1}^{2} \right)
  \epsilon^{abcd}
  \bm{J}_{ab}
  \bm{J}_{cd} -
  2 \left( \alpha_{2}^{2} - \alpha_{1}^{2} \right)
  \epsilon^{abcd}
  \bm{J}_{ab}
  \bm{Z}_{cd} +
  \left( \alpha_{0} \alpha_{2} - \alpha_{1}^{2} \right)
  \epsilon^{abcd}
  \bm{Z}_{ab}
  \bm{Z}_{cd}
\right]
\nonumber \\ & = -
\frac{\epsilon^{abcd}}{\det \left( m_{\alpha \beta} \right)}
\left[
  \alpha
  \bm{Z}_{ab}
  \bm{Z}_{cd} -
  2 \left( \beta - \alpha \right)
  \bm{J}_{ab}
  \bm{Z}_{cd} +
  \left( \beta - \alpha \right)
  \bm{J}_{ab}
  \bm{J}_{cd}
\right].
\label{3.43}
\end{align}

The Casimir operators of the SSEP algebra
obtained in Refs.~\cite{Sor04,Dup05,Sor06,Sor10} are apparently different from the ones shown
in eqs.~(\ref{3.40})--(\ref{3.41}).
The mismatch, however, is only superficial.
Indeed, if we take $c=1$ and $a=i/2$ in eqs.~(2.2) and~(2.3) from Ref.~\cite{Sor06},
we readily get the operators $\bm{C}_{1}$ and $\bm{C}_{2}$ shown in eqs.~(\ref{3.40})--(\ref{3.41}).

Performing the same rescaling and choosing
$\alpha $ $=1$ y $\beta $ $=2$ in $\bm{\bar{C}}_{JZ}$,
we can verify that the Casimir operator $\bm{C}_{3}$ of Ref.~\cite{Sor06}
exactly matches our $\bm{\bar{C}}_{JZ}$ Casimir operator.

\section{A generalized action for Chern--Simons gravity in $2+1$ dimensions}
\label{sec:CS}

In this section we find a rank-two, symmetric invariant tensor for the SSEP algebra
and use it to build the more general action for CS gravity in $2+1$ dimensions.

\subsection{The Invariant Tensor}

It is easy to see that the most general symmetric invariant tensor of rank two for the AdS algebra
in three-dimensional spacetime is given by
(see, e.g., Ref.~\cite{Iza09a})
\begin{align}
\left\langle
  \bm{J}_{ab}
  \bm{J}_{cd}
\right\rangle
& =
\tilde{\mu}_{0}
\left(
  \eta_{ad} \eta_{bc} - \eta_{ac} \eta_{bd}
\right),
\label{ti1} \\
\left\langle
  \bm{J}_{ab}
  \bm{P}_{c}
\right\rangle
& =
\tilde{\mu}_{1}
\epsilon_{abc},
\label{ti2} \\
\left\langle
  \bm{P}_{a}
  \bm{P}_{b}
\right\rangle
& =
\tilde{\mu}_{0}
\eta_{ab},
\label{ti3}
\end{align}
where $\mu_{0}$ and $\mu_{1}$ are arbitrary constants.
Theorem~7.2 from Ref.~\cite{Iza06b} assures us that the only nonzero components
of the corresponding symmetric invariant tensor for the SSEP algebra are
\begin{align}
\left\langle
  \bm{N}_{ab}
  \bm{N}_{cd}
\right\rangle & =
\alpha_{0}
\left(
  \eta_{ad}
  \eta_{bc} -
  \eta_{ac}
  \eta_{bd}
\right),
\label{ti4}
\\
\left\langle
  \bm{L}_{ab}
  \bm{L}_{cd}
\right\rangle & =
\alpha_{2}
\left(
  \eta_{ad}
  \eta_{bc} -
  \eta_{ac}
  \eta_{bd}
\right),
\label{ti6}
\\
\left\langle
  \bm{L}_{ab}
  \bm{L}_{c3}
\right\rangle & =
\alpha_{1}
\epsilon_{abc},
\label{ti8}
\\
\left\langle
  \bm{L}_{a3}
  \bm{L}_{b3}
\right\rangle & =
\alpha_{2}
\eta_{ab},
\label{ti9}
\end{align}
where $\alpha_{0}$, $\alpha_{1}$ and $\alpha_{2}$ are arbitrary constants.

\subsection{Chern--Simons action for the semi-simple extended Poincar\'{e}
algebra in $2+1$ dimensions}

A generic CS Lagrangian in $\left( 2+1 \right)$-dimensional spacetime reads~\cite{Cha89,Cha90,Zan05}
\begin{equation}
L_{\text{CS}}^{2+1} =
2 k
\int_{0}^{1} dt
\left\langle
  \bm{A} \left( t \mathrm{d} \bm{A} + t^{2} \bm{A}^{2} \right)
\right\rangle =
k
\left\langle
  \bm{A} \left( \mathrm{d} \bm{A} + \frac{2}{3} \bm{A}^{2} \right)
\right\rangle,
\label{chs1}
\end{equation}
where $\bm{A}$ is a Lie algebra-valued one-form gauge connection and $k$ is an arbitrary coupling constant.%
\footnote{Wedge product between differential forms is understood throughout. Note that commutators between Lie algebra-valued differential forms carry the expected sign changes, so that, e.g.,
$\left[ \bm{A}, \bm{A} \right] = 2 \bm{A} \bm{A} = 2 \bm{A}^{2}$.}
For the SSEP algebra we may write
\begin{equation}
\bm{A} =
\frac{1}{2}
\varpi^{ab}
\bm{N}_{ab} +
\frac{1}{2}
\omega^{AB}
\bm{L}_{AB} =
\frac{1}{2}
\varpi^{ab}
\bm{N}_{ab} +
\frac{1}{2}
\omega^{ab}
\bm{L}_{ab} +
\omega^{a3}
\bm{L}_{a3}.
\label{chs2}
\end{equation}%
For the sake of convenience, let us define the SSEP-valued one-form gauge fields
\begin{align}
\bm{\varpi} =
\frac{1}{2} \varpi^{ab} \bm{N}_{ab}, \\
\bm{\omega} =
\frac{1}{2} \omega^{ab} \bm{L}_{ab}, \\
\bm{\varphi} =
\omega^{a3} \bm{L}_{a3}.
\end{align}
In terms of these, $\bm{A}$ takes on the simple form
\begin{equation}
\bm{A} = \bm{\varpi} + \bm{\omega} + \bm{\varphi}.
\label{chs3}
\end{equation}%
A straightforward calculation shows that the CS Lagrangian for the SSEP algebra in three-dimensional spacetime may be written as
\begin{align}
L_{\text{SSEP}}^{2+1} & = k
\left\langle
  \bm{\varpi} \mathrm{d} \bm{\varpi} +
  \bm{\varpi} \mathrm{d} \bm{\omega} +
  \bm{\varpi} \mathrm{d} \bm{\varphi} +
  \frac{1}{3} \bm{\varpi}
  \left[
    \bm{\varpi},
    \bm{\varpi}
  \right]
\right\rangle +
\nonumber \\ & + k
\left\langle
  \bm{\omega} \mathrm{d} \bm{\varpi} +
  \bm{\omega} \mathrm{d} \bm{\omega} +
  \bm{\omega} \mathrm{d} \bm{\varphi} +
  \frac{1}{3} \bm{\omega}
  \left[
    \bm{\omega},
    \bm{\omega}
  \right] +
  \frac{2}{3} \bm{\omega}
  \left[
    \bm{\omega},
    \bm{\varphi}
  \right] +
  \frac{1}{3} \bm{\omega}
  \left[
    \bm{\varphi},
    \bm{\varphi}
  \right]
\right\rangle
\nonumber \\ & + k
\left\langle
  \bm{\varphi} \mathrm{d} \bm{\varpi} +
  \bm{\varphi} \mathrm{d} \bm{\omega} +
  \bm{\varphi} \mathrm{d} \bm{\varphi} +
  \frac{1}{3} \bm{\varphi}
  \left[
    \bm{\omega},
    \bm{\omega}
  \right] +
  \frac{2}{3} \bm{\varphi}
  \left[
    \bm{\omega},
    \bm{\varphi}
  \right] +
  \frac{1}{3}
  \bm{\varphi}
  \left[
    \bm{\varphi},
    \bm{\varphi}
  \right]
\right\rangle.
\label{chs5}
\end{align}

The SSEP two-form curvature reads
\begin{align}
\bm{F} & =
\mathrm{d} \bm{A} + \bm{A}^{2}
\nonumber \\ & =
\mathrm{d} \bm{\varpi} +
\mathrm{d} \bm{\omega} +
\mathrm{d} \bm{\varphi} +
\bm{\varpi} \bm{\varpi} +
\bm{\omega} \bm{\omega} +
\bm{\varphi} \bm{\varphi} +
\left[
  \bm{\omega},
  \bm{\varphi}
\right]
\nonumber \\ & =
\left(
  \mathrm{d} \bm{\varpi} +
  \bm{\varpi} \bm{\varpi}
\right) +
\left(
  \mathrm{d} \bm{\omega} +
  \bm{\omega} \bm{\omega}
\right) +
\left(
  \mathrm{d} \bm{\varphi} +
  \bm{\varphi} \bm{\varphi} +
  \left[
    \bm{\omega},
    \bm{\varphi}
  \right]
\right),
\label{chs6}
\end{align}
so that it proves convenient to define the following partial curvatures:
\begin{align}
\tilde{\bm{R}} & =
\mathrm{d} \bm{\varpi} +
\bm{\varpi} \bm{\varpi} =
\mathrm{d} \bm{\varpi} +
\frac{1}{2}
\left[
  \bm{\varpi},
  \bm{\varpi}
\right],
\label{chs7} \\
\bm{R} & =
\mathrm{d} \bm{\omega} +
\bm{\omega} \bm{\omega} =
\mathrm{d} \bm{\omega} +
\frac{1}{2}
\left[
  \bm{\omega},
  \bm{\omega}
\right],
\\
\tilde{\bm{T}} & =
\mathrm{d} \bm{\varphi} +
\bm{\varphi} \bm{\varphi} +
\left[
  \bm{\omega},
  \bm{\varphi}
\right] =
\bm{T} +
\frac{1}{2}
\left[
  \bm{\varphi},
  \bm{\varphi}
\right],
\end{align}
where $\bm{T} = \mathrm{d} \bm{\varphi} + \left[ \bm{\omega}, \bm{\varphi} \right]$.

From the definition of covariant derivative we can write 
\begin{align}
\mathrm{D} \bm{\varpi} & =
\mathrm{d} \bm{\varpi} +
\left[
  \bm{\varpi},
  \bm{\varpi}
\right]
\\
\mathrm{D} \bm{\omega} & =
\mathrm{d} \bm{\omega} +
\left[
  \bm{\omega},
  \bm{\omega}
\right] +
\left[
  \bm{\omega},
  \bm{\varphi}
\right],
\\
\mathrm{D} \bm{\varphi} & =
\mathrm{d} \bm{\varphi} +
\left[
  \bm{\omega},
  \bm{\varphi}
\right] +
\left[
  \bm{\varphi},
  \bm{\varphi}
\right] =
\bm{T} +
\left[
  \bm{\varphi},
  \bm{\varphi}
\right].
\label{chs9}
\end{align}

Using eqs.~(\ref{chs7})--(\ref{chs9}) in~(\ref{chs5}), we get
\begin{align}
L_{\text{SSEP}}^{2+1} & =
\frac{k}{4}
\varpi^{ab}
\left(
  \mathrm{d} \varpi^{cd} +
  \frac{2}{3}
  \nwse{\varpi}{c}{e}
  \varpi^{ed}
\right)
\left\langle
  \bm{N}_{ab}
  \bm{N}_{cd}
\right\rangle +
\frac{k}{4}
\omega^{ab}
\left(
  \mathrm{d} \omega^{cd} +
  \frac{2}{3}
  \nwse{\omega}{c}{e}
  \omega^{ed}
\right)
\left\langle
  \bm{L}_{ab}
  \bm{L}_{cd}
\right\rangle +
\nonumber \\ & + k
\left(
  R^{ab}
  \omega^{c3} -
  \frac{2}{3}
  \omega^{a3}
  \omega^{b3}
  \omega^{c3}
\right)
\left\langle
  \bm{L}_{ab}
  \bm{L}_{c3}
\right\rangle +
k \mathrm{D} \omega^{a3}
\omega^{c3}
\left\langle
  \bm{L}_{a3}
  \bm{L}_{c3}
\right\rangle
- \mathrm{d}
\left(
  \frac{k}{2}
  \omega^{ab}
  \omega^{c3}
\left\langle
  \bm{L}_{ab}
  \bm{L}_{c3}
\right\rangle
\right).
\label{chs10}
\end{align}

Introducing the invariant tensor~(\ref{ti4})--(\ref{ti9}) in eq.~(\ref{chs10}),
we find that the CS action for the SSEP algebra, in the
$\left\{ \bm{N}_{ab}, \bm{L}_{CD} \right\}$
basis, is given by%
\footnote{Here we have absorbed $k$ in the $\alpha_{i}$ constants.}
\begin{align}
S_{\text{SSEP}}^{\left( 2+1 \right)} & =
\int_{M}
\frac{1}{2}
\alpha_{0}
\nwse{\varpi}{a}{c}
\left(
  \mathrm{d} \nwse{\varpi}{c}{a} +
  \frac{2}{3}
  \nwse{\varpi}{c}{d}
  \nwse{\varpi}{d}{a}
\right) +
\alpha_{1}
\epsilon_{abc}
\left(
  R^{ab}
  \omega^{c3} +
  \frac{1}{3}
  \omega^{a3}
  \omega^{b3}
  \omega^{c3}
\right) +
\nonumber \\ & +
\alpha_{2}
\mathrm{D} \omega^{a3}
\swne{\omega}{a}{3} +
\frac{1}{2}
\alpha_{2}
\nwse{\omega}{a}{c}
\left(
  \mathrm{d}
  \nwse{\omega}{c}{a}+
  \frac{2}{3}
  \nwse{\omega}{c}{d}
  \nwse{\omega}{d}{a}
\right) -
\mathrm{d}
\left(
  \frac{\alpha_{1}}{2}
  \epsilon_{abc}
  \omega^{ab}
  \omega^{c3}
\right).
\end{align}

Relabeling $\omega ^{a3} = e^{a}/l$, where $l$ is a length, we may write
\begin{align}
S_{\text{SSEP}}^{\left( 2+1 \right)} & =
\frac{\alpha_{0}}{2}
\int_{M}
\nwse{\varpi}{a}{c}
\left(
  \mathrm{d} \nwse{\varpi}{c}{a} +
  \frac{2}{3}
  \nwse{\varpi}{c}{d}
  \nwse{\varpi}{d}{a}
\right) +
\nonumber \\ & +
\frac{\alpha_{1}}{l}
\left[
  \int_{M}
  \epsilon_{abc}
  \left(
    R^{ab}
    e^{c} +
    \frac{1}{3l^{2}}
    e^{a}
    e^{b}
    e^{c}
  \right) -
  \frac{1}{2}
  \int_{\partial M}
  \epsilon_{abc}
  \omega^{ab}
  e^{c}
\right]
\nonumber \\ & +
\frac{\alpha _{2}}{2}
\int_{M}
\left[
  \nwse{\omega}{a}{c}
  \left(
    \mathrm{d} \nwse{\omega}{c}{a} +
    \frac{2}{3}
    \nwse{\omega}{c}{d}
    \nwse{\omega}{d}{a}
  \right) +
  \frac{2}{l^{2}}
  e_{a}
  T^{a}
\right],
\label{chs12}
\end{align}%
where we have used
$\mathrm{D} \omega^{a3} =
\mathrm{D} e^{a} / l =
T^{a} / l$.
The action in eq.~(\ref{chs12}) is probably the most general action for CS
gravity in $2+1$ dimensions.

\section{Comments}
\label{sec:final}

We have shown that:
(i)~the SSEP algebra
$\mathfrak{so} \left( D-1,1 \right) \oplus \mathfrak{so} \left( D-1,2\right)$
of Refs.~\cite{Sor04,Dup05,Sor06,Sor10} can be obtained from the AdS algebra
$\mathfrak{so} \left( D-1,2 \right)$ via the $S$-expansion procedure~\cite{Iza06b,Iza09a}
with an appropriate semigroup $S$;
(ii)~there exists a prescribed method for computing Casimir operator for $S$-expanded algebras,
which is exemplified through the SSEP algebra; and
(iii)~the above-mentioned $S$-expansion methods allowed us to obtain an invariant tensor for the
SSEP algebra, which in turn permits the construction of
the more general action for CS gravity in $2+1$ dimensions.

The interesting facts here are that the resultant theory corresponds to the sum
of the CS forms associated to the direct sum of
$\mathfrak{so} \left( D-1,1\right) \oplus \mathfrak{so} \left( D-1,2\right)$
of the Lorentz and the AdS Lie algebras.

The action~(\ref{chs12}) includes among its terms:
(i)~a term corresponding to the so-called ``exotic Lagrangian''
for the connection $\bm{\varpi}$, which is invariant under the Lorentz algebra~\cite{Zan05}; and
(ii)~the topological Mielke--Baekler action for three-dimensional gravity (for details, see Ref.~\cite{Mie91}).

\begin{acknowledgments}
P.~S.\ was supported in part by
Direcci\'{o}n de Investigaci\'{o}n, Universidad de Concepci\'{o}n through Grant \# 210.011.053-1.0
and in part by
Fondecyt through Grant \# 1080530.
Three of the authors (O.~F., N.~M. and O.~V.) were supported by grants from the
Comisi\'{o}n Nacional de Investigaci\'{o}n Cient\'{\i}fica y Tecnol\'{o}gica CONICYT and from
Universidad de Concepci\'{o}n, Chile.
J.~D.\ was supported in part by Universidad Arturo Prat.
F.~I.\ and E.~R.\ were supported by the
National Commission for Scientific and Technological Research, Chile, through Fondecyt research grants 11080200 and 11080156, respectively.
\end{acknowledgments}

\bibliographystyle{utphys}
\bibliography{biblio2012}

\end{document}